\documentclass[12pt,reqno]{article}

\usepackage{xcolor}
\usepackage{amssymb}
\usepackage{amsmath}
\usepackage{amsthm}
\usepackage{amsfonts}
\usepackage{graphicx}
\usepackage{float}

\usepackage[colorlinks=true,
linkcolor=webgreen,
filecolor=webbrown,
citecolor=webgreen]{hyperref}

\definecolor{webgreen}{rgb}{0,.5,0}
\definecolor{webbrown}{rgb}{.6,0,0}

\usepackage{color}
\usepackage{fullpage}

\usepackage{graphics}
\usepackage{latexsym}
\usepackage{epsf}

\newcommand{\seqnum}[1]{\href{https://oeis.org/#1}{\rm \underline{#1}}}

\usepackage{pgfplots}
\pgfplotsset{compat=1.18}
\usepackage{enumitem}

\begin{document}


\theoremstyle{plain}
\newtheorem{theorem}{Theorem}
\newtheorem{corollary}[theorem]{Corollary}
\newtheorem{lemma}[theorem]{Lemma}
\newtheorem{proposition}[theorem]{Proposition}

\theoremstyle{definition}
\newtheorem{definition}[theorem]{Definition}
\newtheorem{example}[theorem]{Example}
\newtheorem{conjecture}[theorem]{Conjecture}

\theoremstyle{remark}
\newtheorem{remark}[theorem]{Remark}

\numberwithin{equation}{section}
\numberwithin{theorem}{section}
\numberwithin{figure}{section}

\begin{center}
\vskip 1cm{\LARGE\bf
Re-Ordering the Generalized Hofstadter Functions $G$, $H$, and Beyond\\
}
\vskip 1cm
\large
    Pierre Letouzey\\
    IRIF \\
    Universit\'e Paris Cit\'e, CNRS, Inria\\
    F-75013 Paris\\
    France\\
    \href{mailto: letouzey@irif.fr}{\tt letouzey@irif.fr}\\
\end{center}

\vskip .2 in
\begin{abstract}
    Hofstadter's G function is recursively defined via $G(0)=0$ and then $G(n)=n-G(G(n-1))$. Following Hofstadter, we vary the number $k$ of nested recursive calls in this equation and obtain a family of functions $(F_k)$.
    In a previous article, we proved that this family is ordered
pointwise: for all $k$ and $n$, we have $F_k(n) \le F_{k+1}(n)$.
Here we propose a much shorter proof of this theorem, without
infinite morphic words. The key ingredient for this direct proof
is a new formula expressing iterates of the preimages of $F_k$.
Thanks to this formula, we also prove several conjectures mentioned
in our previous article. In particular we obtain that
$F_k(n) < F_{k+1}(n)$ when $n > \tfrac{1}{2}(k+1)(k+6)$.
  We provide a certified formalization of all these results in the Rocq proof assistant.
\end{abstract}

\section{Introduction}

In a previous article~\cite{LetouzeyShuoSteiner2026}, we provided the first
proof of pointwise monotonicity for a family of generalized Hofstadter
functions. We also expressed several conjectures concerning these
functions, and in particular a conjecture about their strict ordering.
The present article is a follow-up of this last work:
this time we provide a much shorter proof
for the main theorem of our last article as well as first proofs for
our former conjectures. Once again, all these proofs has been
formalized and verified in the Rocq proof assistant, as described
in Section~\ref{Rocq}.

This new article can be read alone, since it contains again all the
the required initial definitions. The basic properties are also
mentioned again briefly, before focusing on the novel content.
Nevertheless, much of the context described in
the former article is still fully relevant but not duplicated here,
in particular the discussion on related
works~\cite[Sect.~1.7]{LetouzeyShuoSteiner2026}.

\subsection{The functions}

As in our previous article~\cite{LetouzeyShuoSteiner2026},
for each integer $k\ge 1$, we consider
the following
recursive function\footnote{In all this work, the
$n$-th exponent of a function denotes its $n$-th iterate.
For instance, $F_2^2(n)$ is $F_2(F_2(n))$, not the square of
$F_2(n)$.}
\[
F_k:\, \mathbb{N} \to \mathbb{N}, \quad n \mapsto \begin{cases}0, & \mbox{if}\ n = 0; \\ n-F_k^k(n-1), & \mbox{otherwise}.\end{cases}
\]
The function $F_k$ is well defined, since one may prove that
$0 \le F_k(n) \le n$ for all $n \ge 0$. This family of functions is
due to Hofstadter \cite[Chap.~5]{GEB}.
In particular, $F_2$ is Hofstadter's
function~$G$ (see OEIS entry~\seqnum{A005206}
\cite{OEIS,GEB,DowneyGriswold84,GaultClint}). Note that
$G(n)=\lfloor (n+1)/\varphi \rfloor$ where $\varphi$ is the
golden ratio $(1+\sqrt{5})/2$. Similarly, $F_3$ is
Hofstadter's function~$H$ (see OEIS \seqnum{A005374} and the recent
study by Shallit~\cite{Shallit2025}).
The generalization to
higher degrees of recursive nesting in the definition of~$F_k$
was already suggested by Hofstadter~\cite{GEB}.
To be complete, the OEIS database already includes
$F_4$ as \seqnum{A005375},
$F_5$ as \seqnum{A005376}, and
$F_6$ as \seqnum{A100721}. On the other hand, we
choose to start this sequence with~$F_1$, where only one recursive call
is done. This $F_1$ function can easily be shown to verify
$F_1(n)=\lfloor (n+1)/2 \rfloor=\lceil n/2 \rceil$.
Note that we could extend the previous definition down to $k=0$ and
obtain a non-recursive function~$F_0$ constantly equal to $1$
except in $0$. But this $F_0$
has little in common with the other~$F_k$ functions, so we exclude
this case $k=0$ most of the time.

\begin{figure}[ht]
\pgfplotsset{width=\linewidth}
\begin{tikzpicture}
  \begin{axis}[
    xmin=0,xmax=30,ymin=0,ymax=25,samples=31,
    xtick = {0,5,...,30},
    ytick = {0,5,...,30},
    legend pos=south east
  ]
 \addplot+[mark=.,color=black,style=dashed,domain=0:30] {x};
 \addlegendentry{id \phantom{=G}}
    \addplot [mark=.,color=blue] coordinates {
(0, 0) (1, 1) (2, 1) (3, 2) (4, 3) (5, 4) (6, 5) (7, 6) (8, 6)
 (9, 7) (10, 7) (11, 8) (12, 9) (13, 9) (14, 10) (15, 11) (16, 12)
 (17, 12) (18, 13) (19, 14) (20, 15) (21, 16) (22, 16) (23, 17)
 (24, 18) (25, 19) (26, 20) (27, 21) (28, 21) (29, 22) (30, 23)};
 \addlegendentry{$F_5 \phantom{=G}$}
    \addplot [mark=.,color=purple] coordinates {
(0, 0) (1, 1) (2, 1) (3, 2) (4, 3) (5, 4) (6, 5) (7, 5) (8, 6)
 (9, 6) (10, 7) (11, 8) (12, 8) (13, 9) (14, 10) (15, 11) (16, 11)
 (17, 12) (18, 13) (19, 14) (20, 15) (21, 15) (22, 16) (23, 17)
 (24, 18) (25, 19) (26, 19) (27, 20) (28, 20) (29, 21) (30, 22)};
 \addlegendentry{$F_4 \phantom{=G}$}
    \addplot [mark=.,color=orange] coordinates {
(0, 0) (1, 1) (2, 1) (3, 2) (4, 3) (5, 4) (6, 4) (7, 5) (8, 5)
 (9, 6) (10, 7) (11, 7) (12, 8) (13, 9) (14, 10) (15, 10) (16, 11)
 (17, 12) (18, 13) (19, 13) (20, 14) (21, 14) (22, 15) (23, 16)
 (24, 17) (25, 17) (26, 18) (27, 18) (28, 19) (29, 20) (30, 20) };
 \addlegendentry{$F_3=H$}
 \addplot+[mark=.,color=red,style=solid,domain=0:30]
 {floor((x+1) * 0.618033988749894903)};
 \addlegendentry{$F_2=G$}
 \addplot+[mark=.,color=black,style=solid,domain=0:30] {ceil(x/2)};
 \addlegendentry{$F_1 \phantom{=G}$};
 \begin{scope}[dashed, every rectangle node/.style={draw,solid,anchor=north},
                       every circle node/.style={fill=black,inner sep=1pt}]
 \draw (7,4) node[circle]{} -- ++ (0,-1.5) node {$N_1=7$};
 \draw (12,8) node[circle]{} -- ++ (0,-3.5) node {$N_2=12$};
 \draw (18,13) node[circle]{} -- ++ (0,-6) node {$N_3=18$};
 \draw (25,19) node[circle]{} -- ++ (0,-9) node {$N_4=25$};
 \end{scope}
 \end{axis}
\end{tikzpicture}
\caption{Plotting $F_1,F_2,\ldots,F_5$.
 Note that $F_k < F_{k+1}$ after $N_k = \tfrac{1}{2}(k+1)(k+6)$.
}
\label{fig:plot}
\end{figure}
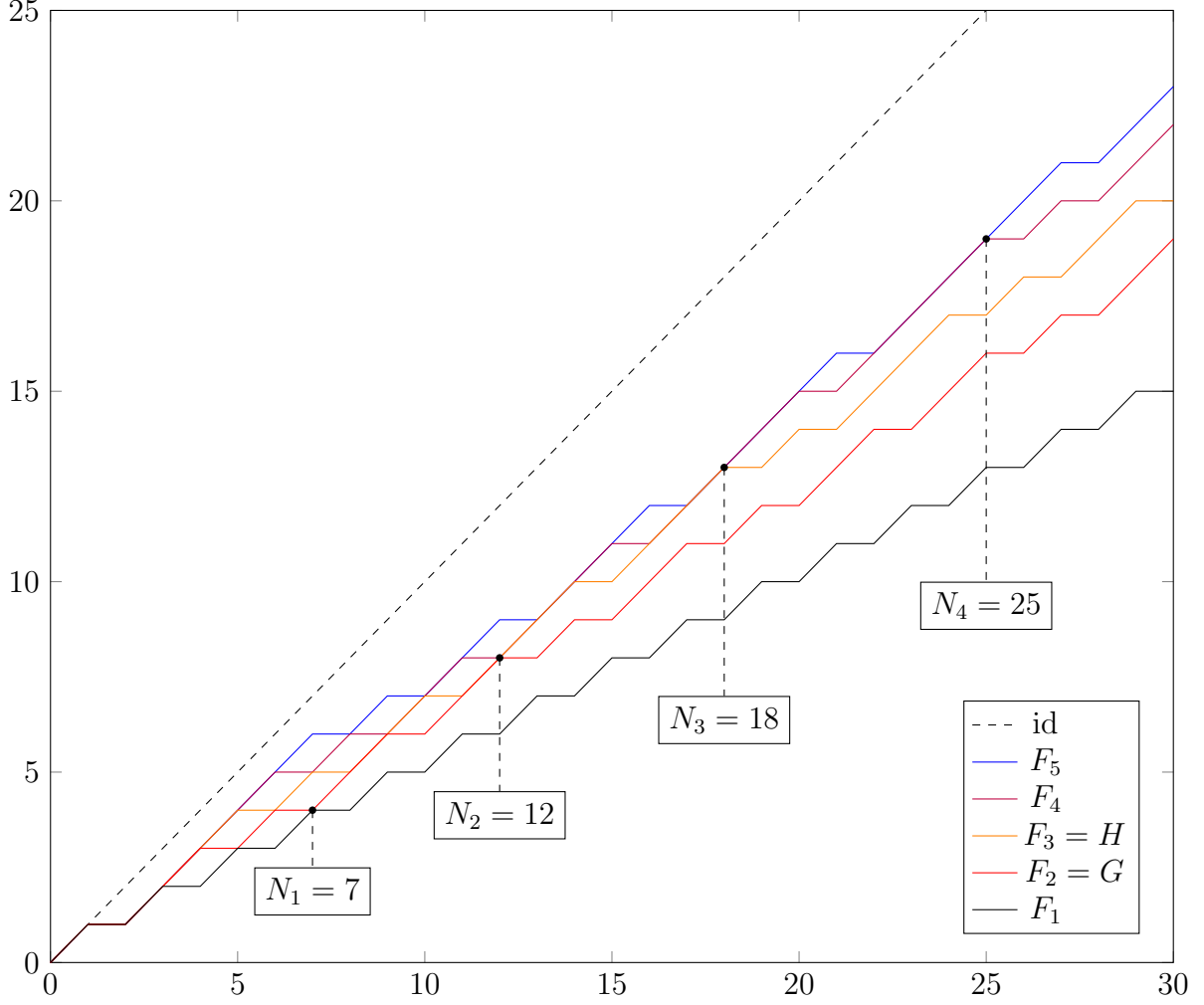

\subsection{Summary of our main results}

The main result of our previous article~\cite{LetouzeyShuoSteiner2026}
is the following theorem, stating that the
$(F_k)$ family of functions is ordered pointwise.

\begin{theorem} \label{t:weak}
For all $k \ge 1$, $n \ge 0$, we have $F_k(n) \le F_{k+1}(n)$.
\end{theorem}

In Section~\ref{shortweak} we revisit this theorem, providing
this time a concise and direct
proof. In particular this new proof does not require
anymore the use of morphic words\footnote{See, for instance,
Allouche and Shallit~\cite{allouche_shallit_2003} for an introduction
to morphic words, also called morphic sequences.}
as in our previous work.
Then we refine further this proof method
and manage to prove in Section~\ref{strictorder} the following
theorem, that was only conjectured in our previous
article~\cite[Sect.~7]{LetouzeyShuoSteiner2026}.

\begin{theorem}
\label{t:strict}
For all $k \ge 1$ and $n > \tfrac{1}{2}(k+1)(k+6)$,
we have $F_k(n) < F_{k+1}(n)$.
\end{theorem}

For the rest of this article, we abbreviate $N_k$ for $\tfrac{1}{2}(k+1)(k+6)$.

As a side note, if we include the degenerated function $F_0$ in our
sequence $(F_k)$, these Theorems~\ref{t:weak} and~\ref{t:strict} also hold
trivially for $k=0$. Our formalization in the Rocq prover includes this
extension (see Section~\ref{Rocq}).

In Section~\ref{formerconj2} we finally prove another statement that we
conjectured in our previous article.

\begin{theorem}
\label{t:Fkkp1}
For all $k \ge 1$, $n\ge 0$, we have $F_k^{k+1}(n) \ge F_{k+1}^{k+2}(n)$.
\end{theorem}

\subsection{The Rocq formalization}
\label{Rocq}

All the proofs presented in this article have been formally certified
using the Rocq proof assistant~\cite{Rocq} (previously known as Coq).
This implementation in Rocq is freely available~\cite{LetouzeyCoqDev}.
Thanks to this Rocq formalization, we ensure precise
definitions and statements while preventing any error during
the corresponding proofs. It can hence serve as a reference for the
interested reader.
In particular, the recommended entry point to read alongside this article~is
\begin{center}
\url{https://github.com/letouzey/hofstadter_g/blob/main/DirectGrowth.v}
\end{center}
The current article tries to be faithful to this formal work
while staying readable by a large audience, at the cost of possible
remaining mistakes during the transcription.

The provided Rocq files can be machine-checked again by any recent
installations of Rocq; see the joint \texttt{README.md} file in the
repository~\cite{LetouzeyCoqDev}.
All the proofs presented in this article
have been preformed within Rocq core logic, with no axioms.
This fact can be checked
via the Rocq command \texttt{Print Assumptions} on our theorems.

\subsection{Summary of the current article}

In Section~\ref{s:basic} we recall some basic properties of
the function $F_k$. Then Section~\ref{newformula} presents
key formulas for the iterates the preimages of $F_k$. This allows us to
provide in Section~\ref{shortweak} a~direct proof of Theorem~\ref{t:weak}
stating that $F_k \le F_{k+1}$.
Then Section~\ref{strictorder} is dedicated to the proof of
Theorem~\ref{t:strict}, i.e., $F_k(n) < F_{k+1}(n)$ for large enough
points $n$.
In Section~\ref{s:diff2}, we explore further this zone where
$F_k(n) < F_{k+1}(n)$, up to the first difference of 2 between $F_k$ and
$F_{k+1}$. Then in Section~\ref{formerconj} we prove two remaining
conjectures.
As a conclusion, we present in Section~\ref{concl} some numerical
experiments about larger differences between $F_k$ and $F_{k+1}$.

\section{Basic properties of function \texorpdfstring{$F_k$}{\it F}}
\label{s:basic}

We recall some basic properties of $F_k$, and in particular
the Proposition~2.1 of our last article~\cite{LetouzeyShuoSteiner2026}.
Here $\partial F_k^j(n)$ denotes the difference
$F_k^j(n+1)-F_k^j(n)$.

\begin{proposition}
\label{p:basicF}
For all $j,k \ge 1$, the function $F_k^j$ satisfies the following
basic properties:
\begin{enumerate}[label=(\alph*)]
\itemsep.5ex
\item $F_k^j(1) = F_k^j(2) = 1$,
\item $F_k^j(n) \ge 1$ whenever $n\ge 1$,
\item $F_k^j(n) < n$ whenever $n\ge 2$,
\item for all $n>0$, $\partial F_k(n) = 1 - \partial F_k^k(n-1)$,
\item hence $\partial F_k^j(n) \in \{0,1\}$ for all $n\ge 0$,
\item the function $F_k^j$ is monotonically increasing and onto (but not one-to-one).
\end{enumerate}
\end{proposition}

Two functions are particularly important for our work, since they express
the least and greatest preimage maps of $F_k$.

\begin{definition}
\label{defIL}
For $k\ge 1$, we define the functions $I_k$ and $L_k$:
\begin{align*}
I_k:\, \mathbb{N} \to \mathbb{N}, \quad n &\mapsto
\begin{cases}0, & \mbox{if}\ n = 0; \\ n+F_k^{k-1}(n-1), &
  \mbox{otherwise};\end{cases} \\
L_k:\, \mathbb{N} \to \mathbb{N}, \quad n &\mapsto n+F_k^{k-1}(n).
\end{align*}
\end{definition}

We already encountered $L_k$ in our previous article
\cite{LetouzeyShuoSteiner2026}, where we defined it as
the length of a substituted prefix of the morphic word $x_k$
(see Proposition~2.3 there). The chosen name comes from this earlier
definition (L for Length). Then the equality of this length with the
expression $n+F_k^{k-1}(n)$ was proved later (see Proposition~4.2 there).
This time, we are able to avoid referring to this morphic word $x_k$.
Hence we favor this alternative choice of definition of $L_k$.
The other function $I_k$ (where $I$ stands for preImage or Inverse)
is actually a
variant of $L_k$, since for $n>0$ we clearly have $I_k(n) = L_k(n-1)+1$.
More generally, $I_k^j(n) = L_k^j(n-1)+1$ for $j\ge 0$ and $n>0$ by immediate
induction over $j$. Some other basic properties of $I_k$ and $L_k$ are
given in the following Proposition.

\begin{proposition} \label{p:IL}
For $k\ge 1$ and $n,j\ge 0$,
\begin{enumerate}[label=(\alph*)]
\itemsep.5ex
 \item the functions $I_k^j$ and $L_k^j$ are strictly monotonically
   increasing over n,
 \item we have $I_k^j(n) \ge n$, with equality only when $n\le 1$ or $j=0$,
 \item we have $L_k^j(n) \ge n$, with equality only when $n=0$ or $j=0$,
 \item finally $0 \le L_k^j(n)-I_k^j(n) \le 2^j-1$.
 \end{enumerate}
\end{proposition}
\begin{proof}
All these statements are obvious when $j=0$.

For (a), the possible increments $I_k(n+1)-I_k(n)$ are
$1+\partial F_k^{k+1}(n-1)$ when $n\ge 1$, or $1$ when $n=0$.
Thanks to Proposition~\ref{p:basicF} these increments are $1$ or $2$.
The increments of $L_k$ are also $1$ or $2$.
Moreover the iteration of a strictly increasing function is also a
strictly increasing function.

For (b), when $j=1$ and
$n>0$ we have $I_k(n)-n = F_k^{k-1}(n-1)$ which is $0$ when $n=1$ and
strictly positive when $n>1$. Moreover $I_k(0) = 0$.
From there, we obtain $I_k^j(n) \ge I_k(n) \ge n$ for $n\ge 0$ and
$j\ge 1$. Moreover for $j\ge 1$, we have
$I_k^j(n)=I_k(n)=n$ only for $n=0$ or $1$.
The proof of the statement (c) is very similar.

Concerning statement (d), for $j=1$
and $n>0$ we have
$L_k(n)-I_k(n) = \partial F_k^{k-1}(n-1) \in \{0,1\}$
thanks to Proposition~\ref{p:basicF} (and $L_k(0)-I_k(0) = 0$).
We extend this to all $j$ by induction.
Let $j\ge 1$ and assume $0 \le L_k^j(n) - I_k^j(n) \le 2^j-1$.
Then
\[
L_k^{j+1}(n) - I_k^{j+1}(n)
     = L_k(L_k^j(n)) - L_k(I_k^j(n)) + L_k(I_k^j(n)) - I_k(I_k^j(n)).
\]
The rightmost difference above is either $0$ or $1$ as already seen
for the case $j=1$.
In the other difference above, the function $L_k$ is applied
on an interval no larger than
$2^j-1$, thanks to the induction hypothesis. Moreover this function
$L_k$ has increments of either $1$ or $2$.
The sum of these two differences is hence between
$0$ and $2(2^j-1)+1 = 2^{j+1}-1$. This concludes the induction.
\end{proof}

Note that the generic upper bound $2^j-1$ given above for $L_k^j(n)-I_k^j(n)$
is reached in the
case $k=1$, but is quite coarse for higher values of $k$.
In particular in Section~\ref{newformula} we obtain instead a bound of $j$
as long as $j\le k$. At least the current bound is fine for $j\le 1$.

We retrieve the important relations between $F_k$ and $L_k$,
also expressed thanks to $I_k$ this time. Note that all intervals
mentioned in this article are discrete intervals of $\mathbb{N}$,
i.e.,  $[a,b] = \{a,a+1,\ldots,b\}$, possibly reduced to a singleton when $a=b$.

\begin{proposition} \label{p:FL}
For $k\ge 1$ and $n,j\ge 0$, we have
\begin{enumerate}[label=(\alph*)]
\itemsep.5ex
\item $F_k^j(L_k^j(n))=n$ and $L_k^j(F_k^j(n))\ge n$ and in particular
 $L_k(F_k(n))\in\{n,n+1\}$,
\item $F_k^j(I_k^j(n))=n$ and $I_k^j(F_k^j(n))\le n$ and in particular
 $I_k(F_k(n))\in\{n-1,n\}$;
\item $F_k^{-j}(\{n\}) = [I_k^j(n),L_k^j(n)]$
\end{enumerate}
\end{proposition}
\begin{proof}
See
\cite[Sect.~3]{LetouzeyShuoSteiner2026}.
Alternatively, here comes a more direct proof, just relying on
Definition~\ref{defIL} and Propositions~\ref{p:basicF}~and~\ref{p:IL}.
For statement (a), first note that
\[
L_k(F_k(n)) = F_k(n) + F_k^k(n) = F_k(n) + (n+1-F_k(n+1)) =
n+1-\partial F_k(n).
\]
By Proposition~\ref{p:basicF}, $\partial F_k(n)\in\{0,1\}$ hence
$L_k(F_k(n))\in\{n,n+1\}$.
We recall that $F_k$ is onto, as mentioned in Proposition~\ref{p:basicF}.
So let us pick $m$ such that $F_k(m)=n$.
We proceed as before, but for $m$ this time:
\[
F_k(L_k(n)) = F_k(L_k(F_k(m))) = F_k(m+1-\partial F_k(m))
\]
Either $\partial F_k(m) = 0$ and hence $F_k(L_k(n)) = F_k(m+1) =
F_k(m) = n$. Or $\partial F_k(m) = 1$ and hence
$F_k(L_k(n)) = F_k(m+1-1) = F_k(m) = n$.

This extends to $F_k^j(L_k^j(n)) = n$ by induction on $j$.
Let $j\ge 0$ and assume that for all $n\ge 0$, $F_k^j(L_k^j(n)) = n$
In particular, this implies
$F_k^j(L_k^j(L_k(n))) = L_k(n)$ and hence
$F_k^{j+1}(L_k^{j+1}(n))) = F_k(L_k(n)) = n$, which concludes this
induction. Similarly, we prove $L_k^j(F_k^j(n)) \ge n$ by induction on
$j$ (thanks to the monotonicity of $L_k$).
This ends the proof of statement (a).

Now we focus on statement (c). This statement is obvious when $j=0$.
Let $j\ge 1$ and let $p$ be $L_k^j(n)+1$.
By (a), we know that $L_k^j(F_k^j(p)) \ge p > L_k^j(n)$.
Hence $F_k^j(p)$ cannot be $n$. Since $F_k^j$ has increments of $0$ or
$1$, and $F_k^j(L_k^j(n)) = n$, we hence obtain that
$F_k^j(L_k^j(n)+1) = n+1$. Since $F_k^j$ is increasing, it cannot
diminish back to $n$ later. So $L_k^j(n)$
is the greatest antecedent of $n$ by $F_k^j$.
This equality $F_k^j(L_k^j(n)+1) = n+1$ can also be
reformulated as $F_k^j(I_k^j(n+1)) = n+1$. Together with the obvious fact
$F_k^j(I_k^j(0))=0$, this establishes that for all $n$, we have
$F_k^j(I_k^j(n)) = n$.
Besides, for $n\neq 0$ we have $F_k^j(I_k^j(n)-1)=F_k^j(L_k^j(n-1))=n-1$.
Hence $I_k^j(n)$ is the least antecedent of $n$ by $F_k^j$.
We finish the proof of statement (c) by noting that all
natural numbers between $I_k^j(n)$ and $L_k^j(n)$ are also antecedents
of $n$ by $F_k^j$, thanks to the monotonicity of $F_k^j$.

For statement (b), we have already proved $F_k^j(I_k^j(n)) = n$ in the
previous paragraph.
For $j\ge 0$, we prove that $I_k^j(F_k^j(n)) \le n$.
For that, note that $n$ is obviously in $F_k^{-j}(\{F_k^j(n)\})$.
So by point (c), we have $I_k^j(F_k^j(n)) \le n$.
We prove similarly that $I_k(F_k(n))\in\{n-1,n\}$.
Indeed, $n$ is in $F_k^{-1}(\{F_k(n)\})$. So by
point (c) for $j=1$, we have $I_k(F_k(n)) \le n \le L_k(F_k(n))$.
Finally, Proposition~\ref{p:IL} indicates that
$I_k(F_k(n))$ is either $L_k(F_k(n))$ or $L_k(F_k(n)) - 1$, so
$I_k(F_k(n))$ is $n$ or $n-1$.
\end{proof}

\begin{corollary}
\label{galoisL}
For all $k\ge 1$ and $j\ge 0$, the functions $F_k^j$ and $L_k^j$
form a \emph{Galois connection}~\cite{Erne1993} between $\mathbb{N}$ and itself
(with $F_k^j$ as \emph{left adjoint} and $L_k^j$ as \emph{right adjoint}).
Indeed, for all $n,m \ge 0$ we have $F_k^j(n) \le m$ iff
$n \le L_k^j(m)$. Said otherwise, $m < F_k^j(n)$ iff $L_k^j(m) < n$.
\end{corollary}
\begin{proof}
See~\cite[Cor.~3.2]{LetouzeyShuoSteiner2026}.
Alternatively,
if $F_k^j(n) \le m$, then $n \le L_k^j(F_k^j(n)) \le L_k^j(m)$
by Proposition~\ref{p:FL} and by monotonicity of $L_k^j$
(Proposition~\ref{p:IL}).
Conversely, $n\le L_k^j(m)$ implies $F_k^j(n) \le F_k^j(L_k^j(m)) = m$.
Finally, the formulation with strict inequalities is obtained by
contrapositive: $m < F_k^j(n)$ iff $F_k^j(n) \nleq m$ iff
$n \nleq L_k^j(m)$ iff $L_k^j(m) < n$.
\end{proof}

\begin{corollary}
\label{galoisI}
For all $k\ge 1$ and $j\ge 0$, the functions $I_k^j$ and $F_k^j$
also form a \emph{Galois connection} between $\mathbb{N}$ and itself
(this time $F_k^j$ is the right adjoint and $I_k^j$ the left adjoint).
Indeed, for all $n,m \ge 0$ we have $I_k^j(n) \le m$ iff
$n \le F_k^j(m)$. Said otherwise, $m < I_k^j(n)$ iff $F_k^j(m) < n$.
\end{corollary}
\begin{proof}
If $I_k^j(n) \le m$, then $n = F_k^j(I_k^j(n)) \le F_k^j(m)$
by Proposition~\ref{p:FL} and by monotonicity of $F_k^j$.
Conversely, $n\le F_k^j(m)$ implies $I_k^j(n) \le I_k^j(F_k^j(m)) \le m$.
Here again, the formulation with strict inequalities is obtained by
contrapositive.
\end{proof}

\section{Explicit formulas for iterated preimages}
\label{newformula}

The following lemma expresses the first iterates of $I_k$ and $L_k$
via simple sums of iterated $F_k$ terms. This is the key ingredient
that was lacking in our previous article~\cite{LetouzeyShuoSteiner2026}.
Its use is essential during the short proof of Theorem~\ref{t:weak} in
Section~\ref{shortweak} and during the first proof of
Theorem~\ref{t:strict} in Section~\ref{strictorder}.
This lemma has been contributed to us by Jeffrey Shallit, who himself
obtained it with the help of ChatGPT 5.5 Extended. This is the only
statement of this article that is due to AI. As the rest of
this work, we checked it with the Rocq proof assistant (see
Section~\ref{Rocq}). We should also mention that ChatGPT suggested
a different proof of Theorem~\ref{t:weak} (but ours is significantly
simpler) as well as a shortcut in our proof of Lemma~\ref{bootstrap} (we
implemented this shortcut).

\begin{lemma}
\label{lemkey}
For $k\ge 1$ and $j \in [0,k]$, we have
\begin{align*}
I_k^j(n) &= n + \sum_{i=k-j}^{k-1} F_k^i(n-1) \quad\quad \text{for}~n>0; \\
L_k^j(n) &= n + \sum_{i=k-j}^{k-1} F_k^i(n) \quad\quad \text{for}~n\ge 0.
\end{align*}
\end{lemma}
\begin{proof}
By induction over $j$. The cases $j=0$ and $j=1$ are clear.
Let $0\le j < k$ and assume for all $n>0$
the previous formulas about $I_k^j(n)$.
For $n>0$, we have that $I_k(n)>0$. The formula at this point gives
\begin{align*}
I_k^j(I_k(n)) &= I_k(n) + \sum_{i=k-j}^{k-1} F_k^i(I_k(n)-1) \\
              &= n + F_k^{k-1}(n-1) +
                 \sum_{i=k-j}^{k-1} F_k^{i-1}(F_k(I_k(n)-1))
\end{align*}
In the sum, we simplify $F_k(I_k(n)-1) = F_k(L_k(n-1)) = n-1$.
After shifting the sum index and incorporating the extra term in the
sum, we obtain
\[
I_k^{j+1}(n) = n + \sum_{i=k-j-1}^{k-1} F_k^i(n-1).
\]
This concludes the induction. The formula about $L_k^j(n)$ can be
proved in the same way.
\end{proof}

The constraint $j\le k$ is crucial for keeping the formulas in such
a simple form, without coefficients in the sum.
For higher values of $j$, we could obtain similar formulas but with
non-trivial coefficients.
As a example,
the formula for $I_k^{k+1}$ is given below (it is used in
Section~\ref{formerconj2}).

\begin{proposition}
\label{pkey}
For $k,n\ge 1$, we have
\[
I_k^{k+1}(n) = 2n-1 + 2F_k^{k-1}(n-1) + \sum_{i=0}^{k-2} F_k^i(n-1).
\]
\end{proposition}
\begin{proof}
Continuing Lemma~\ref{lemkey}, we start with the formula for $j=k$ and use
one more time the same method, i.e., apply this formula at $I_k(n)$:
\[
I_k^{k}(I_k(n)) = I_k(n) + \sum_{i=0}^{k-1} F_k^i(I_k(n)-1).
\]
In the sum, we consider the case $i=0$ separately, obtaining
an extra $I_k(n)-1$, for a total of $2 I_k(n)-1 = 2n-1+2F_k^{k-1}(n-1)$.
The other terms of the sum are handled as in
the proof of Lemma~\ref{lemkey}.
\end{proof}

As an immediate application of Lemma~\ref{lemkey}, we can provide
an improved upper bound for $L_k^j(n)-I_k^j(n)$ when $j\le k$. Instead
of the generic $2^j-1$ given in Proposition~\ref{p:IL}, we obtain
$j$ by subtracting the two formulas of Lemma~\ref{lemkey} and
recalling that $\partial F_k^j$ is at most~$1$.
Thanks to Proposition~\ref{p:FL}, this means that for $j\le k$,
the function $F_k^j$ has at most $j+1$ antecedents for a given value.
Such situations with $j+1$ antecedents are easy to
exhibit, consider for instance the antecedents of $1$.

In a very similar way, we can also consider
the increments $I_k^j(n+1)-I_k^j(n)$ of the function $I_k^j$, and find
these increments to be in $[1,j+1]$ as long as $j\le k$.
The same property also holds for $L_k^j$.
We use this fact later in the particular case $j=k$. In this case
the sum of differences always includes $\partial F_k^0(n-1) = 1$. This
leads to a slightly refined range.

\begin{corollary}
\label{c:dIkk}
For $k,n\ge 1$, the increment $I_k^k(n+1)-I_k^k(n)$ is in the range $[2,k+1]$.
In particular this increment is $2$ when $F_k(n-1) = F_k(n)$ and
it is $k+1$ when $\partial F_k^i(n-1) = 1$ for all $i \in [1,k-1]$.
\end{corollary}

\section{Direct proof of non-strict ordering}
\label{shortweak}
The Lemma~\ref{lemkey} is a powerful tool: from the knowledge of $F_k$
and its iterates at a given point, we deduce $I_k^k$ at the next
point, which in turn controls the behavior of
$F_k$ at a greater location. This allows us to simultaneously establish
$F_k \le F_{k+1}$ and $I_k^k \le I_{k+1}^{k+1}$ in the proof by
induction of
Theorem~\ref{t:weak} given below.

Note that using the least preimage $I_k^k$ instead of
the greatest preimage $L_k^k$
is mostly a matter of taste. The following proofs can be adapted to
$L_k^k$ as well. However, it seems crucial to exploit the $k$-th
iterate of a preimage function of $F_k$. For instance,
a similar proof based on $I_k$ instead of $I_k^k$ does not
seem possible. Indeed, comparing $I_k$ and $I_{k+1}$
is equivalent to comparing $F_k^{k-1}$ and $F_{k+1}^{k}$, which is not
directly related with the comparison of $F_k$ and $F_{k+1}$.

First, we show that $F_k(n) = n-m$ whenever $I_k^k(m) < n \le
I_k^k(m+1)$. More precisely, we split this property in two parts.

\begin{proposition}
\label{p:highlow}
For $k\ge 1$ and $n,m\ge 0$:
\begin{enumerate}[label=(\alph*)]
\itemsep.5ex
\item if $I_k^k(m) < n$ then $F_k(n) \le n-m$;
\item if $n \le I_k^k(m+1)$ then $n-m \le F_k(n)$.
\end{enumerate}
\end{proposition}
\begin{proof}
\begin{enumerate}[label=(\alph*)]
\itemsep.5ex
\item Assume $I_k^k(m) < n$, or equivalently $I_k^k(m) \le n-1$.
  By Corollary~\ref{galoisI}, this is also equivalent to
  $m \le F_k^k(n-1) = n - F_k(n)$.
\item Assume $n \le I_k^k(m+1)$, or equivalently $n \le L_k^k(m)+1$,
 i.e., $n-1 \le L_k^k(m)$. By Corollary~\ref{galoisL}, this is also
 equivalent to $F_k^k(n-1) \le m$. Hence $n-F_k(n) \le m$.
\end{enumerate}
\end{proof}

Then we compare $I_k^k$ and $I_{k+1}^{k+1}$, assuming that
$F_k$ and $F_{k+1}$ (and their iterates) are well ordered just before.

\begin{lemma}
\label{lem:F2I}
For $k \ge 1$ and $n\ge 0$, if $F_k^i(n) \le F_{k+1}^i(n)$ for all
$i\ge 0$
then $I_k^k(n+1) \le I_{k+1}^{k+1}(n+1)$.
\end{lemma}
\begin{proof}
By Lemma~\ref{lemkey} we obtain
\begin{align*}
I_{k+1}^{k+1}(n+1) - I_k^k(n+1)
 &= n+1 + \sum_{i=0}^{k} F_{k+1}^i(n) - (n+1) - \sum_{i=0}^{k-1} F_k^i(n) \\
 &= F_{k+1}^{k}(n) + \sum_{i=0}^{k-1} (F_{k+1}^i(n) - F_k^i(n)).
\end{align*}
All the differences in the sum are non-negative thanks to the
hypothesis. And of course $F_{k+1}^{k}(n)$ is
also non-negative. So $I_k^k(n+1) \le I_{k+1}^{k+1}(n+1)$ as desired.
\end{proof}

Now we are equipped to prove the Theorem~\ref{t:weak} in a direct way.

\begin{proof}[Proof of Theorem~\ref{t:weak}]
Let $k\ge 1$. We prove $F_k(n) \le F_{k+1}(n)$ by strong induction on $n\ge 0$.
First, $F_k(0) \le F_{k+1}(0)$ since both are $0$.
Now let $n\ge 1$ and assume that $F_k \le F_{k+1}$ everywhere in
$[0,n-1]$.
This induction hypothesis extends
to iterates of these functions: for all $i\ge 0$, we have
$F_k^i \le F_{k+1}^i$ on $[0,n-1]$.
This is proved by a routine induction on $i$ where
all the functions involved are
monotonically increasing and smaller than the identity
hence preserve the interval $[0,n-1]$.

We focus on the value $m = F_k^k(n-1)$.
First this $m$ is in
$[0,n-1]$ since $F_k^k$ is smaller than the identity. We hence have
$F_k^i(m) \le F_{k+1}^i(m)$ for all $i\ge 0$.
By Lemma~\ref{lem:F2I} we deduce that $I_k^k(m+1) \le I_{k+1}^{k+1}(m+1)$.

Furthermore, the chosen value of $m$ ensures that
$I_k^k(m) < n \le I_k^k(m+1)$
Indeed, we have $m \le F_k^k(n-1) < m+1$.
Hence by Corollary~\ref{galoisI} we have
$I_k^k(m) \le n-1 < I_k^k(m+1)$.

Finally $F_k(n) \le n - m \le F_{k+1}(n)$. Indeed the left side
is a direct use of Proposition~\ref{p:highlow}(a). And the right side
comes from $n \le I_k^k(m+1) \le I_{k+1}^{k+1}(m+1)$ combined
with Proposition~\ref{p:highlow}(b). This concludes the proof by induction.
\end{proof}

In the previous proof, note that the induction hypothesis is used
only on $m$ or less. This point $m$ is also equal to $n-F_k(n)$
(so in fact $F_k(n) = n - m$ instead of the mere inequality used
above). This $m$ coincides with $n-1$ for
$n=1,2$ but becomes quite smaller than $n-1$ afterwards.
For instance, we clearly have $m\le n-2$ when $n>2$.
Moreover, this $m$ could have being chosen as
$n-F_{k+1}(n) = F_{k+1}^{k+1}(n-1)$ instead of $F_k^k(n-1)$ with only minor
adaptations to the rest of the proof. Despite being smaller, this
alternative choice does not improve the proof. This situation occurs
again in Section~\ref{sub:strictorder}, but there the smaller alternative
is actually more cumbersome that using $F_k^k(n-1)$.

From Theorem~\ref{t:weak}, we retrieve some immediate consequences
(also see Theorem~7.4 and Equation~(7.3) in
\cite{LetouzeyShuoSteiner2026}).

\begin{corollary}
\label{c:fsgrows}
For $k\ge 1$ and $n,i\ge 0$, we have $F_k^i(n) \le F_{k+1}^i(n)$.
\end{corollary}
\begin{proof}
By induction on $i$, just as in the beginning of the previous proof.
\end{proof}

\begin{corollary}
For $k\ge 1$ and $n\ge 0$, we have $I_k^k(n) \le I_{k+1}^{k+1}(n)$
as well as $L_k^k(n) \le L_{k+1}^{k+1}(n)$.
\end{corollary}
\begin{proof}
The case $n = 0$ is clear. Otherwise we use the Lemma~\ref{lem:F2I},
whose precondition is proved by Corollary~\ref{c:fsgrows}.
Finally we recall that $L_k^k(n) = I_k^k(n+1)-1$ and similarly for
$L_{k+1}^{k+1}$.
\end{proof}

\section{Strict ordering of \texorpdfstring{$F_k$ and $F_{k+1}$}{F}}
\label{strictorder}

This Section is dedicated to the proof of Theorem~\ref{t:strict}.
After studying the initial values of $F_k$ and $I_k^k$
in Section~\ref{s:aux}, we provide an ad-hoc early comparison
 of $I_k^k$ and $I_{k+1}^{k+1}$ in Section~\ref{s:bootstrap}.
This allows us to prove Theorem~\ref{t:strict} by induction in
Section~\ref{sub:strictorder}.

\subsection{Initial behavior \texorpdfstring{of $F_k$ and $I_k^k$}{}}
\label{s:aux}

Now we establish a few elementary properties of the
initial values of $F_k$ and $I_k^k$ that are needed in the rest of
this Section.

For now on, we abbreviate $T_p$ for the $p$-th
triangular number $\tfrac{1}{2}p(p+1)$. Note that the value
$N_k = \tfrac{1}{2}(k+1)(k+6)$ introduced earlier can also be
expressed as $T_{k+3}-3$.

First, the function $F_k$ admits an initial zone $[2,k+2]$ where it
satisfies $F_k(n) = n-1$. We call it the ``linear'' initial zone.
More generally, we express $F_k^j$ and $I_k^k$ in this zone.

\begin{proposition}
\label{p:FIinit}
Let $k\ge 1$.
\begin{enumerate}[label=(\alph*)]
\item For $n\in [1,k+2]$ and $j\ge 0$, we have $F_k^j(n) = \max(1, n - j)$;
\item for $n\in [2,k+2]$ we have $I_k^k(n) = T_{n-1} + k+1$.
\end{enumerate}
\end{proposition}
\begin{proof}
\begin{enumerate}[label=(\alph*)]
\item
By induction on $n\in [1,k+2]$.
The base case is clear: $F_k^i(1) = 1 = \max(1,1-j)$.
Now let $n\in [2,k+2]$ and assume $F_k^i(n-1) = \max(1,n-1-i)$ for all $i$.
We prove $F_k^i(n) = \max(1,n-i)$ for all $i$.
The case $i=0$ is obvious. Now we consider $i=1$.
In that case, $F_k(n) = n-F_k^k(n-1) = n-\max(1,n-1-k)=n-1 = \max(1,n-1)$
since $n\le k+2$ here.
Finally, for $i>1$,
 $F_k^i(n) = F_k^{i-1}(F_k(n)) = F_k^{i-1}(n-1) = \max(1,n-1-(i-1))$.
This ends the current induction.
\item
For $n\in[2,k+2]$, we use Lemma~\ref{lemkey} and the previous point.
This gives
\[
I_{k}^{k}(n) = n + \sum_{i=0}^{k-1} F_k^i(n-1)
             = n + \sum_{i=0}^{k-1} \max(1,n-1-i).
\]
When $n< k+2$, this last sum can be split in the following way.
\begin{align*}
I_{k}^{k}(n) &= n + \sum_{i=0}^{n-2} (n-1-i) + \sum_{i=n-1}^{k-1} 1\\
             &= n + T_{n-1} + ((k-1)-(n-1)+1) \\
             &= k+1 + T_{n-1}.
\end{align*}
Instead of adapting this computation to the case $n=k+2$,
we use Corollary~\ref{c:dIkk} in this case, with an increment of $k+1$
compared to $I_k^k(k+1)$. Indeed $\partial F_k^i(k) = 1$
for all $i\in[1,k-1]$ thanks to point (a).
Since the previous computation provides $I_k^k(k+1) = k+1+T_{k} = T_{k+1}$,
we obtain the desired equality $I_k^k(k+2) = T_{k+1}+k+1$.
\end{enumerate}
\end{proof}

The linear zone we just considered is followed by a ``triangular''
zone $[k+3, N_k]$ where
the behavior of $F_k$ is governed by triangular numbers. First we
express a few particular cases at the left of this triangular zone.

\begin{proposition}
\label{p:FIkp}
For $k\ge 1$ we have the following particular cases.
\begin{align*}
  F_k(k+3) &= k+1; \\
  F_k(k+4) &= k+2; \\
  I_k^k(k+3) &= T_{k+3} - 3 = N_k; \\
  I_k^k(k+4) &= T_{k+3} - 1 = 2 + N_k; \\
  I_k^k(k+5) &= T_{k+4} - 4 = N_{k+1} - 1.
\end{align*}
\end{proposition}
\begin{proof}
We compute $F_k(k+3) = k+3-F_k^k(k+2)=k+3-\max(1,k+2-k)=k+1$
by Proposition~\ref{p:FIinit}(a).
In particular, this means that $F_k(k+3) = F_k(k+2)$ and
more generally $F_k^i(k+3) = F_k^i(k+2)$ for all $i>0$.
Then we exploit the fact that $F_k$ cannot be flat twice in a row.
More precisely, we have
$\partial F_k(k+3) = 1 - \partial F_k^k(k+2) = 1$
by Proposition~\ref{p:basicF}, so $F_k(k+4) = k+2$.

For the equations concerning $I_k^k$, we proceed once again by
incremental differences thanks to Corollary~\ref{c:dIkk}.
\begin{itemize}
\item We established $I_k^k(k+2) = T_{k+1} +k+1$ in the previous Proposition.
\item $I_k^k(k+3) = I_k^k(k+2) + k+1$. Indeed, by
  Proposition~\ref{p:FIinit}(a) we have $\partial F_k^i(k+1) =1$
  for all $i\in[1,k-1]$
\item $I_k^k(k+4) = I_k^k(k+3) + 2$. Indeed $F_k(k+3) = k+1 = F_k(k+2)$.
\item $I_k^k(k+5) = I_k^k(k+4) + k+1$.
 Indeed for $i\in[1,k-1]$ we have $F_k^i(k+4) = F_k^{i-1}(k+2) = k+2-(i-1)$
 and similarly $F_k^i(k+3) = F_k^{i-1}(k+1) = k+1-(i-1)$.
\end{itemize}
From there, the desired equations are obtained by routine reasoning
on triangular numbers.
\end{proof}

We recall
that the triangular numbers form a strictly increasing sequence
starting at $T_0=0$, so for all numbers $m\ge 0$, there is a unique $p$ such
that $T_p \le m < T_{p+1}$. We call this $p$ the ``triangular index'' of $m$.
In the interval $[k+3,N_k]$, the
value of $F_k(n)$ is governed by the triangular index of $n-k-2$.

\begin{proposition}
\label{p:Ftri}
Let $k\ge 1$ and $n\in[k+2,N_k]$.
\begin{enumerate}[label=(\alph*)]
\item For all $p$ such that $T_p \le n-k-2$, we have $F_k(n) \le n-p-1$;
\item for all $p$ such that $n-k-2 < T_{p+1}$, we have $F_k(n) \ge
  n-p-1$;
\item consequently $F_k(n) = n-p-1$ where $p$ is the triangular index
  of $n-k-2$.
\end{enumerate}
\end{proposition}
\begin{proof}
Let $k\ge 1$ and $n\in[k+2,N_k]$.
\begin{enumerate}[label=(\alph*)]
\item Let $p$ such that $T_p \le n-k-2$.
If $p = 0$, we indeed have $F_k(n) \le n-1$ for $n\ge 2$ by
Proposition~\ref{p:basicF}(c). Now we assume $p>0$.
By Proposition~\ref{p:highlow}(a), to prove that $F_k(n) \le n - p -1$,
it is actually sufficient to establish $I_k^k(p+1) < n$.
Note that we have $p < k+2$ since
$T_p \le n-k-2 \le N_k - k - 2 = T_{k+3}-k-5 < T_{k+2}$.
Hence we are entitled to use Proposition~\ref{p:FIinit}(b) and obtain
$I_k^k(p+1) = T_{p}+k+1 < n$.
\item Let $p$ such that $n-k-2 < T_{p+1}$.
By Proposition~\ref{p:highlow}(b), to prove that $F_k(n) \ge n - p -1$,
it is actually sufficient to establish $n \le I_k^k(p+2)$.
If $p \le k$, we can use Proposition~\ref{p:FIinit}(b) to obtain
$I_k^k(p+2) = T_{p+1}+k+1 \ge n$.
Otherwise $p \ge k+1$ and in this case we have
$I_k^k(p+2) \ge I_k^k(k+3) = N_k \ge n$ by Propositions~\ref{p:IL}(a)
and~\ref{p:FIkp}.
\item Direct combination of the two previous points.
\end{enumerate}
\end{proof}

As a side note, two other descriptions of this triangular zone are
worth mentioning.
First, Hofstadter~\cite{GEB} noticed that $F_k$ can be presented
thanks to an infinite tree obtained from a pattern specific to $k$,
and then labeled by successive numbers in a breadth-first way.
For a node in this tree, the function $F_k$ transforms its label
into the label of its parent node. On this topic, also see
Celaya and Ruskey~\cite{CelayaRuskey2013} and our earlier
work~\cite{Letouzey15}.  Here, the triangular zone
corresponds to the initial part of the tree where just one binary node
occurs per row.

Secondly, $F_k$ is deeply related with a non-standard numeration
system that we describe in a separate work~\cite{Letouzey2025},
following Meek and Van Rees~\cite{MeekVanRees84}.
In this numeration system, the weights of the
digits are given by a generalization of the
Fibonacci sequence. For Fibonacci, we recursively add
the last two terms of the sequence to obtain the next one.
Here we recursively add $k$-apart terms in the sequence:
$A_{k,n} = A_{k,n-1} + A_{k,n-k}$ combined with the initial terms
$A_{k,n} = n+1$ when $n<k$. When a number $n$ is written in this
numeration system, the value of $F_k(n)$ is
obtained by shifting down the digits of $n$ by one position, with a particular
treatment of the least-significant digit. We recall that the
the triangular zone we just explored ends with $N_k = \tfrac{1}{2}(k+1)(k+6)$.
This $N_k$ happens to be the first number to
necessitate three non-zero digits in this numeration system.
All the other numbers in the linear and triangular zones can be
written with at most two non-zero digits. Thanks to this fact,
it is possible to obtain alternative proofs for our results about
the triangular zone, including the tedious Lemma~\ref{bootstrap}
given in the next Section. But these alternative proofs are not
significantly simpler than the ones presented here.
Finally, these interpretations of $F_k$ based on
trees or numeration systems do not seem to provide direct proofs of
Theorems~\ref{t:weak} nor~\ref{t:strict}.

From Proposition~\ref{p:Ftri},
it is actually possible to express the general iterates $F_k^j$ in the
triangular zone.
In particular $F_k^j$ is alternatively flat there (after triangular
numbers or immediately before) or growing several times (further away
from the previous triangular number). The exact details are rather
intricate. Fortunately, we can avoid them for the rest of this article.
Instead, we just propose
some simple lower bounds for $F_k^k$ and $F_k^{k-1}$ in the triangular zone.

\begin{proposition}
\label{p:Fkkm1}
Let $k\ge 1$ and $n\in[k+1,N_k-1]$.
For all $p$ such that $T_p \le n-k-1$, we have $F_k^k(n) \ge p+1$ and
$F_k^{k-1}(n) \ge p+2$.
\end{proposition}
\begin{proof}
For the first inequality,
$F_k^k(n)$ is expressed as $n+1-F_k(n+1)$ while
Proposition~\ref{p:Ftri}(a) leads to $F_k(n+1)\le (n+1)-p-1$.

We obtain the second inequality by proving additionally that
$F_k^k(n) < F_k^{k-1}(n)$. This is ensured by
Proposition~\ref{p:basicF}(c) as soon as $F_k^{k-1}(n) \ge 2$.
But we have
$n \ge k+1$, and hence $F_k^{k-1}(n) \ge F_k^{k-1}(k+1)$ by monotonicity.
Finally $F_k^{k-1}(k+1) = \max(1,k+1-(k-1)) = 2$ by
Proposition~\ref{p:FIinit}(a).
\end{proof}

Finally, we describe the early differences between $F_{k+1}$ and $F_k$.

\begin{proposition}
\label{p:Fdifflow}
Let $k\ge 1$.
\begin{enumerate}[label=(\alph*)]
\item For $n \in [0,k+2]$, we have $F_{k+1}(n) = F_k(n)$.
\item For $n \in [k+3,N_k]$, we have $F_{k+1}(n) = F_k(n)+1$
when $n-k-2$ is a triangular number, otherwise $F_{k+1}(n) = F_k(n)$.
\item In particular
\begin{align*}
  F_{k+1}(k+3) &= F_k(k+3)+1 = k+2; \\
  F_{k+1}(k+4) &= F_k(k+4) = k+2; \\
  F_{k+1}(N_k) &= F_k(N_k) = N_{k-1}+1.
\end{align*}
\end{enumerate}
\end{proposition}
\begin{proof}
Let $k\ge 1$.
\begin{enumerate}[label=(\alph*)]
\item We clearly have $F_{k+1}(0) = 0 = F_k(0)$ and
$F_{k+1}(1) = 1 = F_k(1)$, while
for $n\in [2,k+2]$ we have
$F_{k+1}(n) = n-1 = F_{k}(n)$ by Proposition~\ref{p:FIinit}(a).
\item
Let $n \in [k+3,N_k]$.
There exists $p$ such that $T_p \le n-k-2 < T_{p+1}$.
By Proposition~\ref{p:Ftri}, this means that $F_k(n) = n-p-1$.

If $n-k-2$ is exactly a triangular number, it can only be $T_p$.
Moreover $(k+1)+2 \le n \le N_{k+1}$ and $T_{p-1}\le n-(k+1)-2 < T_p$
so another instance of Proposition~\ref{p:Ftri} gives
$F_{k+1}(n) = n - (p-1)-1 = F_k(n)+1$.

Suppose now that $n-k-2$ is not a triangular number.
This implies $T_p < n-k-2$ and hence $T_p \le n-(k+1)-2 < T_{p+1}$,
so finally $F_{k+1}(n) = n - p - 1 = F_k(n)$.

\item For $n =k+3$, we have $n-k-2 = 1 = T_1$.
For $n=k+4$, we have $n-k-2 = 2$ which is not a triangular number.
Finally for $n=N_k$ we have $N_k-k-2 = T_{k+2}-2$. For $k\ge 1$, this
value cannot be a triangular number. The actual values for these
instances of $F_{k+1}$ and $F_k$ are obtained by
Propositions~\ref{p:FIkp} and~\ref{p:Ftri}.
\end{enumerate}
\end{proof}

\subsection{Ad-hoc early comparison \texorpdfstring{of $I_k^k$ and $I_{k+1}^{k+1}$}{}}
\label{s:bootstrap}

This section is dedicated to the comparison of $I_k^k$ and
$I_{k+1}^{k+1}$ in the triangular zone $[k+3,N_k]$.
More precisely, we establish the following lemma, which is essential
to ``bootstrap'' the proof by induction of Theorem~\ref{t:strict} given in
Section~\ref{sub:strictorder}.

\begin{lemma}
\label{bootstrap}
For $k\ge 1$ and $n \in [k+3, N_k]$,
we have $I_k^k(n+1) \le I_{k+1}^{k+1}(n)$.
\end{lemma}

Note that this result actually holds for all $n\ge k+3$, as shown
in Section~\ref{sub:strictorder}. But the generic argument given later
for $n>N_k$ does not hold in the triangular zone, where
we do not have $F_k<F_{k+1}$ everywhere
(see Proposition~\ref{p:Fdifflow}).
Instead, we proceed in a rather ad-hoc manner here: we compare
once again the sums obtained by Lemma~\ref{lemkey}, but this time
the initial terms in these sums may induce a ``debt'' that have to be
compensated for by the extra term $F_{k+1}^k(n-1)$ in the rightmost sum.
Before giving the proof of Lemma~\ref{bootstrap}, we start by a few
auxiliary results for comparing respective terms in these sums.

\begin{proposition}
\label{p:boot1}
Let $k,n\ge 1$ and $p\ge 0$.
\begin{enumerate}[label=(\alph*)]
\item We have $F_k^p(n) \le F_{k+1}^p(n-1) +1$;
\item if $F_k^p(n) \le F_{k+1}^p(n-1)$ then for all $q\ge p$ we also
  have $F_k^q(n) \le F_{k+1}^q(n-1)$.
\end{enumerate}
\end{proposition}
\begin{proof}
\begin{enumerate}[label=(\alph*)]
\item By Corollary~\ref{c:fsgrows}, we have
 $F_k^p(n) \le F_{k+1}^p(n)$. Moreover $F_k^p(n) - F_k^p(n-1) \in \{0,1\}$ by
  Proposition~\ref{p:basicF}.
\item From $F_k^p(n) \le F_{k+1}^p(n-1)$, we obtain
\[
 F_k^q(n) = F_k^{q-p}(F_k^p(n)) \le F_k^{q-p}(F_{k+1}^p(n-1)) \le
 F_{k+1}^{q-p}(F_{k+1}^p(n-1)) = F_{k+1}^q(n-1)
\]
by monotonicity of $F_k^{q-p}$ (Proposition~\ref{p:basicF}) then by
Corollary~\ref{c:fsgrows}.
\end{enumerate}
\end{proof}

\begin{proposition}
\label{p:boot2}
Let $k\ge 1$ and $n \in [k+3,N_k]$ and $p \ge 0$.
If $T_p \le n-k-2 \le T_p +1$ then we have $F_k(n) = F_{k+1}(n-1)$.
\end{proposition}
\begin{proof}
If $p \le 1$, we have either $n = k+3$ or $n=k+4$.
In these cases,
$F_k(k+3) = k+1 = F_{k+1}(k+2)$ and
$F_k(k+4) = k+2 = F_{k+1}(k+3)$ by Propositions~\ref{p:FIinit}
and~\ref{p:FIkp}.
Now we assume $p>1$.
By Proposition~\ref{p:Ftri}, we have $F_k(n) = n - p - 1$.
We can also obtain that $T_{p-1} \le (n-1)-(k+1)-2 < T_p$ hence
another instance of Proposition~\ref{p:Ftri} provides
$F_{k+1}(n-1) = (n-1) - (p-1)-1 = n-p-1$.
\end{proof}

The more $n-k-2$ differs from the previous triangular number,
the more we have to iterate $F_k$ and $F_{k+1}$ to get a favorable situation.

\begin{proposition}
\label{p:boot3}
Let $k\ge 1$ and $n \in [k+3,N_k]$ and $0 < q \le p$.
If $T_p \le n-k-2 \le T_p + q$ then we have $F_k^q(n) \le F_{k+1}^q(n-1)$.
\end{proposition}
\begin{proof}
Let $k\ge 1$.
We prove by induction over $q\ge 1$ that for all
$n\in[k+3,N_k]$ and
$p\ge q$ such that $T_p \le n-k-2 \le T_p + q$, then
$F_k^q(n) \le F_{k+1}^q(n-1)$.
The case $q=1$ is directly implied by the previous Proposition.
Now let $q\ge 2$ and assume the desired property holds at $q-1$.
Let $n\in[k+3,N_k]$ and
$p\ge q$ such that $T_p \le n-k-2 \le T_p + q$.
First, if $F_k(n) \le F_{k+1}(n-1)$, then Proposition~\ref{p:boot1}(b)
directly gives the desired inequality. Now we assume
$F_{k+1}(n-1) < F_k(n)$.
Thanks to Proposition~\ref{p:boot1}(a), this implies
$F_{k+1}(n-1) = F_k(n)-1$.
We abbreviate $F_k(n)$ as $m$.
Our goal $F_k^q(n) \le F_{k+1}^q(n-1)$
can be reformulated at $F_k^{q-1}(m) \le F_{k+1}^{q-1}(m-1)$, which is
an instance of the induction hypothesis for $q-1$. Indeed, we show below
that $m$ is still in $[k+3,N_k]$ and
$T_{p-1} \le m-k-2 \le T_{p-1}+q-1$.
For that, first note that $m = F_k(n) = n-p-1$ by Proposition~\ref{p:Ftri}
(here the additional constraint $q\le p$ is used to obtain $n-k-2 < T_{p+1}$).
Moreover, if we suppose that $n-k-2 \le T_p + 1$, then the previous
Proposition provides $F_k(n) = F_{k+1}(n-1)$, which is impossible
here. We hence have $T_p + 1 < n-k-2$.
\begin{itemize}
\item Let us show that $k+3 \le m \le N_k$. The right part is obvious
since $m = n-p-1 \le n \le N_k$. Now $k+3 \le n-p-1$ is equivalent to
$p+2 \le n-k-2$.
If $p\ge 3$, we indeed have
$p+2 \le T_p \le n-k-2$. Otherwise $p$ cannot be less than $q$, so
$p=q=2$. In this case, $n-k-2$ can only be $T_2+2 = 5$, which is greater
than $p+2=4$.
\item The fact that $T_{p-1} \le m-k-2 \le T_{p-1}+q-1$ is a direct
consequence of $m=n-p-1$ and $T_{p-1} = T_p - p$ and
$T_p+1 < n-k-2 \le T_p + q$.
\end{itemize}
Now that these constraints on $m$ are established, we obtain
$F_k^{q-1}(m) \le F_{k+1}^{q-1}(m-1)$ as planned thanks to the induction
hypothesis. This concludes the current induction.
\end{proof}

Now we are ready to prove Lemma~\ref{bootstrap}.

\begin{proof}[Proof of Lemma~\ref{bootstrap}]
Let $k\ge 1$ and $n\in [k+3,N_k]$. We prove $I_k^k(n+1) \le I_{k+1}^{k+1}(n)$.
Let $p$ be such that $T_p \le n-k-2 < T_{p+1}$.
Note that $p$ is at least $1$, otherwise we would have
$n-k-2 < T_1 =1$ hence $n<k+3$.
Also note that $p \le k+1$, since $T_p \le n-k-2 \le N_k-k-2 < T_{k+2}$.
Let $q$ be $\max(1,n-k-2-T_p)$. This ensures $0<q\le p$.
Moreover we have $q\le k$, otherwise we would have
$q = p = k+1$ hence $n = T_{k+1} +2k+3 > N_k$.
By Lemma~\ref{lemkey} we obtain
\begin{align*}
I_{k+1}^{k+1}(n) - I_k^k(n+1)
 &= n + \sum_{i=0}^{k} F_{k+1}^i(n-1) - (n+1) - \sum_{i=0}^{k-1} F_k^i(n) \\
 &= F_{k+1}^{k}(n-1) - 2 + \sum_{i=1}^{k-1} (F_{k+1}^i(n-1) - F_k^i(n)).
\end{align*}
For $i \ge q$, the differences $(F_{k+1}^i(n-1) - F_k^i(n))$ are
all non-negative. Indeed we have
$F_k^q(n)\le F_{k+1}^q(n-1)$ by Proposition~\ref{p:boot3}.
This propagates to all $i\ge q$ by Proposition~\ref{p:boot1}(b).
Now we consider the low indices $1\le i \le q-1$. The corresponding
differences are all at least
$-1$ thanks to Proposition~\ref{p:boot1}(a). There are $q-1$ such low
indices, so the whole sum is at least $1-q$.
Meanwhile, $F_{k+1}^k(n-1)$ is at least $p+1$, thanks to
Proposition~\ref{p:Fkkm1} used with $k+1$ and $n-1$ and $p-1$.
Indeed $n-1 \in [k+2,N_{k+1}-1]$ here, since $N_k\le N_{k+1}$.
And $T_p \le n-k-2$ hence $T_{p-1} \le T_p-1 \le (n-1)-(k+1)-1$.
Finally $I_{k+1}^{k+1}(n) - I_k^k(n+1)$ is at least $(p+1)-2+(1-q) \ge 0$.
\end{proof}

The previous proof could be adapted to obtain the strict inequality
$I_k^k(n+1) < I_{k+1}^{k+1}(n)$ except when $n=k+3$, but this refinement
does not seem useful later.

\subsection{The main induction for strict ordering}
\label{sub:strictorder}

This Section presents the core of the proof of Theorem~\ref{t:strict},
namely that $F_k(n) < F_{k+1}(n)$ when $n$ exceeds $N_k$.
The proof given here is similar to the one of
Theorem~\ref{t:weak} in Section~\ref{shortweak}. Here again, we relate
$F_k$ and $F_{k+1}$ on one side and $I_k^k$ and $I_{k+1}^{k+1}$ on the
other side. The Lemma~\ref{lemkey} is once again crucial in one direction.
But this time, a strict order such as $I_k^k < I_{k+1}^{k+1}$
at some point is not
sufficient to guarantee a strict order $F_k < F_{k+1}$ at some other
point. Instead, we need a stronger property between $I_k^k$ and
$I_{k+1}^{k+1}$, more precisely that
$I_k^k(n+1) \le I_{k+1}^{k+1}(n)$ for some position $n$.
Moreover, this time the desired properties only hold after a
transitory phase, leading to a quite delicate proof. In particular
we rely on the ``bootstrap'' Lemma~\ref{bootstrap} proved in the
previous section.

First, the next lemma is an adaptation of Lemma~\ref{lem:F2I}
in the case of strict monotonicity.

\begin{lemma}
\label{lem:F2Ioverlap}
For $k\ge 1$ and $n\ge 0$, if $F_k(n) < F_{k+1}(n)$ then
  $I_k^k(n+1) \le I_{k+1}^{k+1}(n)$.
\end{lemma}
\begin{proof}
Let $n$ be such that $F_k(n) < F_{k+1}(n)$.
This implies that $n$ is at least $k+3$ (see Proposition~\ref{p:Fdifflow}).
In particular, $n$ is strictly positive. By Lemma~\ref{lemkey} we obtain
\begin{align*}
I_{k+1}^{k+1}(n) - I_k^k(n+1)
 &= n + \sum_{i=0}^{k} F_{k+1}^i(n-1) - (n+1) - \sum_{i=0}^{k-1} F_k^i(n) \\
 &= F_{k+1}^{k}(n-1) - 2 + \sum_{i=1}^{k-1} (F_{k+1}^i(n-1) - F_k^i(n)).
\end{align*}
Since $n-1\ge k+2$ we have
$F_{k+1}^k(n-1) \ge F_{k+1}^k(k+2) = \max(1,k+2-k) = 2$ by
Proposition~\ref{p:FIinit}(a).
Moreover, all the differences in the remaining sum are non-negative.
Indeed, we recall from Proposition~\ref{p:basicF} that
$\partial F_{k+1}(n-1) \in \{0,1\}$. Since $F_k(n) < F_{k+1}(n)$, this
implies that $F_k(n) \le F_{k+1}(n-1)$.
More generally, for $i\ge 1$, we have
\[
F_k^i(n) = F_k^{i-1}(F_k(n)) \le F_{k+1}^{i-1}(F_k(n)) \le
F_{k+1}^{i-1}(F_{k+1}(n-1)) = F_{k+1}^i(n-1)
\]
by Corollary~\ref{c:fsgrows} and then monotonicity of $F_{k+1}^{i-1}$.
\end{proof}

Now we finally prove Theorem~\ref{t:strict}.

\begin{proof}[Proof of Theorem~\ref{t:strict}]
Let $k\ge 1$. We prove $F_k(n) < F_{k+1}(n)$ by strong induction on $n> N_k$.
Let $n>N_k$ and assume that $F_k < F_{k+1}$ everywhere in
$[1+N_k,n-1]$.
We consider $m = F_k^k(n-1)$, which is also $n-F_k(n)$.
The chosen value of $m$ ensures that $n \le I_k^k(m+1)$.
Indeed, we have $F_k^k(n-1) < m+1$,
so by Corollary~\ref{galoisI} we have $n-1 < I_k^k(m+1)$.
By Proposition~\ref{p:FIkp} we obtain
$I_k^k(k+3) = N_k < n \le I_k^k(m+1)$.
The strict monotonicity of $I_k^k$ (see Proposition~\ref{p:IL})
implies $k+3 < m+1$, i.e., $k+3 \le m$.
Now we prove $F_k(n) < F_{k+1}(n)$, or equivalently
$n-(m-1) \le F_{k+1}(n)$ (since $m = n - F_k(n)$).
By Proposition~\ref{p:highlow}(b), this
inequality holds as soon as $n \le I_{k+1}^{k+1}(m)$.
Since $n \le I_k^k(m+1)$, it is sufficient to establish
$I_k^k(m+1) \le I_{k+1}^{k+1}(m)$.
This fact is provided by the ``bootstrap''
Lemma~\ref{bootstrap} when $m \le N_k$. Besides, when $m > N_k$
the induction hypothesis holds for $m$ and can be combined with
Lemma~\ref{lem:F2Ioverlap} to conclude.
\end{proof}

In the previous proof, the strong induction is presented without
explicit base case. Nevertheless, for $n=1+N_k$ and the values
immediately after, the corresponding $m = F_k^k(n-1)$ is less than
or equal to $N_k$, leading to the use of the ``bootstrap''
Lemma~\ref{bootstrap} and preventing any circularity.

Besides, the value $m$ used during the proof could also be chosen as
$n-F_{k+1}(n) = F_{k+1}^{k+1}(n-1)$ instead of $F_k^k(n-1)$.
But this time, the adaptations to the proof are more consequent than
in Section~\ref{shortweak}:
this smaller alternative $m$ might be smaller than $k+3$, so a few extra
cases have be to dealt with.

\begin{corollary}
For $k\ge 1$ and $n\ge k+3$, we have $I_k^k(n+1) \le I_{k+1}^{k+1}(n)$.
\end{corollary}
\begin{proof}
If $n\le N_k$, the result is given by
Lemma~\ref{bootstrap}. Otherwise $n > N_k$ and
we apply Lemma~\ref{lem:F2Ioverlap} whose condition is ensured
by Theorem~\ref{t:strict}.
\end{proof}

\begin{corollary}
For $k\ge 1$ and $n > N_k$, we have $F_k(n) \le F_{k+1}(n-1)$.
\end{corollary}
\begin{proof}
By Theorem~\ref{t:strict}, we have $F_k(n) < F_{k+1}(n)$ here.
Moreover $F_{k+1}(n) - F_{k+1}(n-1) \in \{0,1\}$ by Proposition~\ref{p:basicF}.
\end{proof}

The following corollary is mentioned as a conjecture in our previous
article~\cite[Sect.~7]{LetouzeyShuoSteiner2026}.

\begin{corollary}
For $k\ge 1$ and $n \neq 1$, we have $F_k(n) < F_{k+1}(n+1)$.
\end{corollary}
\begin{proof}
\begin{itemize}
\item The case $n = 0$ is obvious.
\item If $n \in [2,k+2]$, we have $F_k(n) = n-1 < n = F_{k+1}(n+1)$
by Proposition~\ref{p:FIinit}(a).
\item For $n \in [k+3,N_k]$, we consider $p$ such that
$T_p \le n-k-2 < T_{p+1}$.
Then $F_k(n) = n-p-1$ by Proposition~\ref{p:Ftri}.
Similarly, we have $F_{k+1}(n+1) = n+1-p-1$ by
Proposition~\ref{p:Ftri}, since $n+1 \in [k+4,N_k+1]$ and $N_k+1 \le N_{k+1}$
and $T_p \le (n+1)-(k+1)-2 < T_{p+1}$.
\item
Finally, if $n > N_k$, we have $F_k(n) < F_{k+1}(n) \le F_{k+1}(n+1)$ by
Theorem~\ref{t:strict} and then monotonicity of $F_{k+1}$.
\end{itemize}
\end{proof}

\section{First difference of 2 \texorpdfstring{between $F_k$ and $F_{k+1}$}{}}
\label{s:diff2}

From Theorem~\ref{t:strict}, we know that $F_k(n) < F_{k+1}(n)$ when
$n>N_k$. In this section we extend Proposition~\ref{p:Fdifflow} by
studying $F_{k+1}(n)-F_k(n)$ immediately
after $N_k$. This difference remains exactly $1$ as long as
$N_k < n < N_{k+1}$ and reaches $2$ for the first time in $N_{k+1}$.

\begin{proposition}
For $k\ge 1$ and $n$ such that $N_k < n < N_{k+1}$, we have
 $F_{k+1}(n) = 1 + F_k(n)$.
\end{proposition}
\begin{proof}
Consider such $k$ and $n$.
By Theorem~\ref{t:strict} we know that $F_k(n) < F_{k+1}(n)$.
Hence it is sufficient to prove that $F_{k+1}(n) \le 1 + F_{k}(n)$.
We pose $m = F_k^k(n-1)$, which is also $n-F_k(n)$.
The chosen value of $m$ ensures that $I_k^k(m) < n \le I_k^k(m+1)$
by Corollary~\ref{galoisI}.
This implies that $k+3 \le m \le k+4$.
Indeed, $I_k^k(k+3) = N_k < n \le I_k^k(m+1)$ by
Proposition~\ref{p:FIkp}. Since $I_k^k$ is strictly increasing, we
hence have $k+3 < m+1$. Similarly,
$I_k^k(m) < n \le N_{k+1}-1 = I_k^k(k+5)$ by
Proposition~\ref{p:FIkp}, and hence $m < k+5$.

Moreover the inequality $F_{k+1}(n) \le 1 + F_k(n)$ can
also be stated as $F_{k+1}(n) \le n -(m-1)$.
By Proposition~\ref{p:highlow}, it is sufficient to prove
$I_{k+1}^{k+1}(m-1) < n$, which is here a consequence of
$I_{k+1}^{k+1}(m-1) \le I_k^k(m)$.
We prove this last inequality in the two possibles
cases $m=k+3$ and $m=k+4$.
If $m=k+3$, we have $I_k^k(k+3) = N_k$ by Proposition~\ref{p:FIkp}
and $I_{k+1}^{k+1}(k+2) = T_{k+1} + k+2 = N_k -k$ by
Proposition~\ref{p:FIinit}.
If $m=k+4$, we have $I_k^k(k+4) = N_k+2$ by Proposition~\ref{p:FIkp}
and $I_{k+1}^{k+1}(k+3) = T_{k+2} + k+2 = N_k +2$ by
Proposition~\ref{p:FIinit}.
\end{proof}

\begin{proposition}
\label{p:diff2}
For $k\ge 1$, we have $F_{k+1}(N_{k+1}) = N_k + 1 = F_k(N_{k+1})+2$.
\end{proposition}
\begin{proof}
The left equality is implied by Proposition~\ref{p:Fdifflow}(c) for
$k+1$.
Now we prove that $F_k(N_{k+1}) = N_k-1$ or equivalently that
$F_k(N_{k+1}) = N_{k+1} - (k+5)$.
Thanks to Proposition~\ref{p:highlow}, it is hence sufficient to
establish $I_k^k(k+5) < N_{k+1} \le I_k^k(k+6)$.
These inequalities are valid since $I_k^k(k+5) = N_{k+1} - 1$
by Proposition~\ref{p:FIkp},
 while $I_k^k(k+6)$ is at least $I_k^k(k+5)+2$ thanks to
Corollary~\ref{c:dIkk}.
\end{proof}

\section{Proving two remaining conjectures}
\label{formerconj}

In this section, we prove two properties that were mentioned as
conjectures in our previous article~\cite[Sect.~7]{LetouzeyShuoSteiner2026}.

\subsection{Strict ordering of \texorpdfstring{$L_{k+1}$ and $L_k$}{L}}

Here we prove that $L_{k+1}(n) < L_k(n)$ when $n$ is large enough.
It is a relatively direct consequence of Theorem~\ref{t:strict}.
First, we express some more equations about $F_k$ around $N_k$.

\begin{proposition}
\label{p:Fagain}
Let $k\ge 1$. We have
\begin{align*}
 F_k(N_k-1) &= N_{k-1}; \\
 F_k(N_k) &= N_{k-1}+1; \\
 F_k(N_k+1) &= N_{k-1}+1; \\
 F_k(N_k+2) &= N_{k-1}+2.
\end{align*}
\end{proposition}
\begin{proof}
The equation $F_k(N_k) = N_{k-1}+1$ already appears in
Proposition~\ref{p:Fdifflow}(c).
More generally, we can reformulate the desired equations into
$F_k(n) = n - m$ with respectively $(n,m) = (N_k-1,k+2)$
then $(n,m) = (N_k, k+2)$ then $(n,m) = (N_k+1, k+3)$ then
$(n,m) = (N_k+2,k+3)$. Indeed, note that $N_k = N_{k-1}+k+3$.
By Proposition~\ref{p:highlow},
the equation $F_k(n) = n - m$ holds as soon as
$I_k^k(m) < n \le I_k^k(m+1)$.
We recall from Propositions~\ref{p:FIinit} and~\ref{p:FIkp} that
$I_k^k(k+2) = T_{k+1}+k+1 = N_{k-1}+2$ and
$I_k^k(k+3) = N_k$ and $I_k^k(k+4) = N_k + 2$.
From there, the necessary inequalities are proved by routine reasoning
on triangular numbers.
\end{proof}

\begin{theorem}
For $k\ge 1$ and $n > N_{k-1}$, we have $L_{k+1}(n) < L_k(n)$.
When $k>1$, we have equality at $N_{k-1}$, namely
$L_{k+1}(N_{k-1}) = N_k - 1 = L_k(N_{k-1})$.
\end{theorem}
\begin{proof}
We start by proving the equality of $L_{k}$ and $L_{k+1}$ at
$N_{k-1}$. Let $k>1$. By proposition~\ref{p:FL}, we know that
$L_k(N_{k-1})$ is the greatest value $m$ satisfying
$F_k(m) = N_{k-1}$. Thanks to the previous proposition~\ref{p:Fagain}
and the monotonicity of $F_k$, we conclude $L_k(N_{k-1}) = N_k-1$.
Moreover, Proposition~\ref{p:Fdifflow} states that
$F_{k+1}(n) = F_k(n)$ for $n \le N_k$ such that $n-k-2$ is not a
triangular number. In particular, we have $F_{k+1}(n) = F_k(n)$
for $n = N_k$ and $n = N_k-1$ (the latter case requires $k>1$ instead
of the usual $k>0$). As a consequence, the greatest value $m$ such that
$F_{k+1}(m) = N_{k-1}$ is also $N_k-1$. Hence $L_{k+1}(N_{k-1}) = N_k-1$.

Now we prove $L_{k+1}(n) < L_k(n)$ when $k\ge 1$ and $n > N_{k-1}$.
By Corollary~\ref{galoisL},
this is equivalent to $n < F_{k+1}(L_k(n))$.
We reformulate the first occurrence of $n$ as
$F_k(L_k(n))$ by Proposition~\ref{p:FL}.
The obtained inequation $F_k(L_k(n)) < F_{k+1}(L_k(n))$ is a instance of
Theorem~\ref{t:strict}, provided that $N_k < L_k(n)$.
Since $n > N_{k-1}$ and $L_k$ is monotonically increasing,
it is sufficient to prove that $N_k < L_k(N_{k-1}+1)$.
But $L_k(N_{k-1}+1) = N_{k}+1$. Indeed,
the previous proposition~\ref{p:Fagain} indicates that $N_{k}+1$ is
the greatest $m$ such that $F_k(m) = N_{k-1}+1$.
\end{proof}

Since $I_k(n) = L_k(n-1)+1$ for $k,n\ge 1$, we can reformulate the
previous theorem for $I_{k+1}$ and $I_k$.

\begin{corollary}
For $k\ge 1$ and $n > N_{k-1}+1$, we have $I_{k+1}(n) < I_k(n)$.
\end{corollary}

\subsection{Weak ordering of \texorpdfstring{$F_k^{k+1}$ and
    $F_{k+1}^{k+2}$}{some iterates of F}}
\label{formerconj2}

Here we prove Theorem~\ref{t:Fkkp1}, which compares $F_{k+1}^{k+2}$
and $F_k^{k+1}(n)$. In particular, it confirms that the conjectured
arrows in the Figure~7.1 of our earlier article~\cite{LetouzeyShuoSteiner2026}
can be turned into proved relations.

\begin{proposition}
\label{p2F}
For $k\ge 1$ and $n,p\ge 0$, we have $F_k^p(n) \le 2 F_{k}^{p+1}(n)$.
\end{proposition}
\begin{proof}
Let $k\ge 1$ and $n,p\ge 0$. First,
$F_k^p(n) \le 2 F_{k}^{p+1}(n)$ is equivalent to $m \le 2 F_k(m)$
where $m = F_k^p(n)$. We prove $m \le 2 F_k(m)$ for all $m$.
For that, it is handy to consider the $L_1$ function which is the
greatest preimage of $F_1$. We recall that $L_1(m) = m + F_1^0(m) = 2m$ and
$m \le L_1(F_1(m))$ by Proposition~\ref{p:FL}. Said otherwise,
$m \le 2 F_1(m)$. Finally $F_1(m) \le F_{k}(m)$ by repeated use of
Theorem~\ref{t:weak}.
\end{proof}

\begin{proposition}
\label{pIkSk}
For $k\ge 1$ and $n\ge 0$, we have $I_k^{k+1}(n) \le I_{k+1}^{k+2}(n)$.
\end{proposition}
\begin{proof}
The case $n=0$ is obvious. For $n>0$ we use Proposition~\ref{pkey} twice.
\begin{align*}
I_{k+1}^{k+2}(n) - I_k^{k+1}(n)
 &= 2n-1 + 2 F_{k+1}^{k}(n-1) + \sum_{i=0}^{k-1} F_{k+1}^i(n-1) \\
 &\quad - \big(2n-1 + 2 F_{k}^{k-1}(n-1) + \sum_{i=0}^{k-2}
   F_k^i(n-1)\big) \\
  &= (2 F_{k+1}^{k}(n-1) - F_k^{k-1}(n-1))
     + \sum_{i=0}^{k-1} (F_{k+1}^i(n-1) - F_k^i(n-1)).
\end{align*}
All the differences in the sum are non-negative
thanks to Corollary~\ref{c:fsgrows}.
Finally $(2 F_{k+1}^{k}(n-1) - F_k^{k-1}(n-1))$ is also non-negative
since $F_k^{k-1}(n-1) \le F_{k+1}^{k-1}(n-1)$ by Corollary~\ref{c:fsgrows}
and $F_{k+1}^{k-1}(n-1) \le 2 F_{k+1}^k(n-1)$ by Proposition~\ref{p2F}.
\end{proof}

\begin{proof}[Proof of Theorem~\ref{t:Fkkp1}]
Let $k\ge 1$ and $n\ge 0$. We prove $F_{k+1}^{k+2}(n) \le F_k^{k+1}(n)$.
By Corollary~\ref{galoisI}, this is equivalent to
$I_{k}^{k+1}(F_{k+1}^{k+2}(n)) \le n$.
We prove it by Propositions~\ref{pIkSk} and~\ref{p:FL}:
\[
I_{k}^{k+1}(F_{k+1}^{k+2}(n)) \le I_{k+1}^{k+2}(F_{k+1}^{k+2}(n)) \le
n.
\]
\end{proof}

\section{Conclusion}
\label{concl}

In the current article, we proved several properties
of functions $F_k$ that we conjectured previously. Moreover we vastly
improved the proof of Theorem~\ref{t:weak}.
If we note $d_k(n)$ the difference $F_{k+1}(n)-F_k(n)$, we can
rephrase Theorem~\ref{t:weak} as $d_k(n)\ge 0$.
Besides, the last occurrence of a difference $d_k(n)$ of $0$
is at $n = N_k = \tfrac{1}{2}(k+1)(k+6)$ (see Proposition~\ref{p:Fdifflow}
and Theorem~\ref{t:strict}), while the first occurrence of $1$
is at $n = k+3$ (Proposition~\ref{p:Fdifflow})
and the first occurrence of $2$ is at $n = N_{k+1}$
(Proposition~\ref{p:diff2}).

From there, one may wonder whether these results could be generalized to
an arbitrary difference~$d$ instead of just $0$ or $1$ or $2$.
First note that for a fixed $k$, the difference $d_k(n)$ becomes
arbitrarily large when $n$ grows, as shown in the the asymptotic study
of our previous article~\cite[Sect.~6]{LetouzeyShuoSteiner2026}.
More specifically, when $n \to +\infty$ we have
$d_k(n)/n \to (\alpha_{k+1}-\alpha_k) > 0$ where
$\alpha_k$ is the positive zero of $X^k+X-1$.
Furthermore, we studied in a separate work~\cite{Letouzey2025}
the \emph{discrepancy}
$F_k(n)-\alpha_k\,n$. This discrepancy is finite for
$k\le 4$, logarithmic for $k=5$ and $O(n^a)$ for $k\ge 6$ and some
exponent $0<a<1$ depending on $k$. See also Shallit~\cite{Shallit2025}
for the case $k=3$. Thanks to that, we can provide some rough bounds
for the zone where a difference $d_k(n)$ of $d$ may occur.
For instance, we established that
$|F_3(n)-\alpha_3\,n|<1$ and $|F_4(n)-\alpha_4\,n|<2$.
As a consequence, we have
$|d_3(n) - (\alpha_4-\alpha_3)\,n|< 3$.
Hence a point $n$ such that $d_3(n)=d$ must satisfy
$(d-3)/(\alpha_4-\alpha_3) < n < (d+3)/(\alpha_4-\alpha_3)$.
However, these bounds are far from giving exactly the first and last
occurrences of a given difference $d$. For greater values of $k$, a
similar analysis is possible, but this time the unbounded discrepancies
lead to bounds that are harder to express and probably less precise.

As a side remark, the difference $d_k(n)$ evolves by steps of $-1$ or $0$
or $1$ when $n$ grows. Indeed,
$\partial F_k(n) \in \{0,1\}$ and similarly for $F_{k+1}$.
So $d_k(n)$ follows a discrete version of the intermediate value
theorem~\cite{Johnsonbaugh01011998}.
Combined with $d_k(0)=0$ and $d_k(n)\ge 0$ and the limit of $d_k$
mentioned previously, this ensures that differences $d_k(n)$ are
strictly smaller that $d$ before the first occurrence of $d$ and
strictly greater that $d$ after the last occurrence of $d$.

As an experiment, we computed
some of these $d_k(n)$ (in practice for $k\le 200$ and
$n\le 100{,}000$).
According to this experiment,
for a given $d\ge 3$ and $k$ sufficiently large, the smallest $n$ for
which $d_k(n) = d$ appears to be
\begin{equation}
   n = \frac{k(k+5)}{2} + d\,k + 2 + \frac{(d+1)(d+2)(d+3)}{6}
\end{equation}
For instance, this formula seems to hold for $d=3$ and $k\ge 4$,
or for $d=4$ and $k\ge 8$, or for $d=5$ and $k\ge 13$.
For $d=2$, the value $N_{k+1}$ obtained in Proposition~\ref{p:diff2}
is close to this formula, but the final term $(d+1)(d+2)(d+3)/6 = 10$ should
be replaced in this case by $5$.

During this experiment,
we also tried to determine the largest $n$ for which $d_k(n) = d$.
This computation only provides partial answers, since $d$ could possibly
occur again after the tested zone. In practice, the differences $d_k(n)$
seem pretty regular and we restrict ourselves to values of $d$
that are much smaller that the highest observed difference.
For a given $d\ge 1$ and $k$ sufficiently large, the largest $n$
for which $d_k(n) = d$ appears to be given by a polynomial
$P_d(k)$ of degree $3$ whose leading coefficient is $\tfrac{1}{6}$.
For instance:
\begin{align}
 P_1(k) &= \frac{k^3}{6} + 3 k^2 + (8-\frac{1}{6})k + 4 \quad
 \text{for~} k\ge 3; \\
 P_2(k) &= \frac{k^3}{6} + \frac{7}{2} k^2 + (\frac{23}{2}-\frac{1}{6})k
 + 8 \quad \text{for~} k\ge 3; \\
 P_3(k) &= \frac{k^3}{6} + \frac{7}{2} k^2 + (\frac{25}{2}-\frac{1}{6})k
 + 13 \quad \text{for~} k\ge 3; \\
 P_4(k) &= \frac{k^3}{6} + 4 k^2 + (17-\frac{1}{6})k + 21
  \quad \text{for~} k\ge 2;\\
 &\ldots\nonumber\\
 P_{10}(k) &= \frac{k^3}{6} + 5 k^2 + (28-\frac{1}{6})k + 73
  \quad \text{for~} k\ge 15.
\end{align}
Note that for $d=0$, we have $P_0(k) = N_k$ of degree 2
(see Theorem~\ref{t:weak} and Proposition~\ref{p:Fdifflow}). But this case
$d=0$ seems to be an exception compared with the aforementioned polynomials
of degree 3 for $d\ge 1$.

It is currently unclear whether the polynomials given above could be
expressed in an unified way.
Moreover, proving these experimental formulas seems out of reach for now.
Finally, it is rather unsatisfactory that these formulas cannot
handle the low values of~$k$. Indeed, the most popular $F_k$ are
arguably $F_2$ and $F_3$, i.e., $G$ and $H$ in Hofstadter's
terminology. This can be hinted by the lengths of their OEIS pages
\seqnum{A005206} and \seqnum{A005374} compared to, say,
\seqnum{A005375} for $F_4$.

\section{Acknowledgments}

The author thank Jeffrey Shallit and Wolfgang Steiner and Shuo Li
for their great help during this work.

\bibliographystyle{jis}
\bibliography{letouzey}
\end{document}